\documentclass[a4paper,12pt]{article}
\usepackage{geometry}
 \usepackage{setspace}
\usepackage{amssymb}
\usepackage{amsthm}
\usepackage{amsmath}
\usepackage{graphicx}
\usepackage{color}
\usepackage{relsize}

\newlength{\figurewidth}
\newlength{\figureheight}
\newtheorem{property}{Proposition}[section]
\newtheorem{definition}{Definition}[section]
\newtheorem{theorem}{Theorem}[section]
\newtheorem{lemma}{Lemma}[section]
\newtheorem{corollary}{Corollary}[section]

\newtheorem{assumption}{Assumption}

\begin{document}
	\setcounter{page}{1}
	
	\title{No-arbitrage implies power-law market impact and\\
	 rough volatility}

	\author{
		Paul Jusselin\footnote{paul.jusselin@polytechnique.edu}~~and Mathieu Rosenbaum\footnote{mathieu.rosenbaum@polytechnique.edu} \\
		\'Ecole Polytechnique, CMAP
		}

	\maketitle
	
	\begin{abstract}
\noindent	Market impact is the link between the volume of a (large) order and the price move during and after the execution of this order. We show that under no-arbitrage assumption, the market impact function can only be of power-law type. Furthermore, we prove that this implies that the macroscopic price is diffusive with rough volatility, with a one-to-one correspondence between the exponent of the impact function and the Hurst parameter of the volatility. Hence we simply explain the universal rough behavior of the volatility as a consequence of the no-arbitrage property. From a mathematical viewpoint, our study relies in particular on new results about hyper-rough stochastic Volterra equations. 
	\end{abstract}


	\section{Introduction}
	It is now well-admitted that volatility is rough. This stylized fact first established in \cite{GJR14} and confirmed in \cite{BLP16,LMPR17} means that the (log-)volatility process of an asset essentially behaves as a fractional Brownian motion (fBm for short) with Hurst parameter of order $0.1$. Recall that a fBm $(W^H_t)_{t\geq 0}$ with Hurst parameter $H\in(0,1)$ is a Gaussian process that can be written under the Mandelbrot-van Ness representation as
	\begin{equation*}
	W^H_t = \int_{-\infty}^{0}\big((t-s)^{H-\frac{1}{2}}-(-s)^{\frac{1}{2}}\big)\mathrm{d}B_s + \int_{0}^{t}(t-s)^{H-\frac{1}{2}}\mathrm{d}B_s,
	\end{equation*}
	with $(B_t)_{t\geq 0}$ a classical Brownian motion.	For any $\varepsilon>0$, the sample paths of $(W^H_t)_{t\geq 0}$ are almost surely $H-\varepsilon$ H\"older (and not $H$ H\"older). Therefore the trajectories are very rough when $H$ is small.\\
	
	Various rough volatility models have been recently introduced in the literature, notably in the purpose of risk management of derivatives, such as the rough Heston model of \cite{EER16} where the asset price $(P_t)_{t\geq 0}$ satisfies\\
	\begin{equation*}
		\frac{\mathrm{d}P_t}{P_t} = \sqrt{V_t}\big(\rho\mathrm{d}B^2_t + \sqrt{1-\rho^2}\mathrm{d}B_t^{1}\big)
	\end{equation*}
	with
	\begin{equation}
	\label{eq:intro_rough_vol}
		V_t = V_0 + \frac{\lambda}{\Gamma(H+\frac{1}{2})}\int_{0}^{t}(t-s)^{H-\frac{1}{2}}\big(\theta(t)-V_s\big)\mathrm{d}s +\frac{\nu}{\Gamma(H+\frac{1}{2})} \int_{0}^{t}(t-s)^{H-\frac{1}{2}}\sqrt{V_s}\mathrm{d}B^1_s,
	\end{equation}
	where $B_1$ and $B_2$ are independent Brownian motions, $\lambda$ and $\nu$ two positive constants, $\theta$ a deterministic non-negative function and $\rho\in(-1,1)$ a correlation factor. The particular interest of this model is that, as for the classical Heston model, semi-explicit pricing and hedging formulas can be obtained, see \cite{EER16,EER17}.\\
	
A puzzling question is the origin of the universal rough volatility property of financial assets. A first explanation is proposed in \cite{EEFR16}. In this work, the authors place themselves in a highly endogenous market, meaning that most orders are sent in reaction to other orders and without economic motivation. They show that in this context, the widely used trading practice of metaorders splitting (see below for definition of a metaorder) leads to the rough Heston dynamic \eqref{eq:intro_rough_vol} for the macroscopic price. However, this result is found using a quite specific parametric model for the high frequency price.\\
	
In this paper, we wish to obtain a fundamental explanation underlying the rough volatility property. In fact we prove that rough volatility is simply a consequence of the no-arbitrage principle together with the existence of market impact.\\
	
Market impact is the fact that on average, a buy order moves the price up and a sell order moves the price down. The impact of a single order being very difficult to assess, one usually considers large sets of orders split by brokers, so-called metaorders. Empirical studies of market impact have shown that for a buy metaorder (and symmetrically for a sell metaorder) market impact can be decomposed in two phases: a transient phase with a concave rise of the price during the metaorder execution, and a decay phase, where the price decreases towards a long-term level after the execution is completed.\\
	
Let us consider a buy (say) metaorder and let $(q_t)_{t\geq 0}$ be the cumulative volume of this metaorder executed between the initial time $0$ and time $t$. The market impact function of this metaorder is defined as
	\begin{equation*}
	MI(t) = \mathbb{E}[P^{(q_s)_{s\leq t}}_t - P_0],
	\end{equation*}
	where we put the superscript $(q_s)_{s\leq t}$ on $P$ to insist on the fact that the price dynamic depends on the execution process of the metaorder. Of course the
above formula only makes sense in a model where $P^{(q_s)_{s\leq t}}_t$ is a well-defined stochastic process, as will be the case in the next sections.\\ 
	
The permanent market impact (PMI for short) of this metaorder is given by the quantity
	\begin{equation*}
		PMI =  \underset{t\rightarrow+\infty}{\lim}MI(t).
	\end{equation*}
	It is shown in \cite{G10, HS04} that the absence of price manipulation on a market\footnote{A price manipulation is a round-trip (strategy starting and finishing with null inventory) whose expected cost is negative.} implies that the permanent market impact is proportional to the total volume of the metaorder. In particular it does not depend on the metaorder execution strategy. This linear permanent market impact property has consequences on the price dynamics. Indeed, assuming moreover that the price $P$ is a martingale, it is shown in \cite{J15} that up to a negligible martingale term,
	\begin{equation}
	\label{eq:price_process_def}
		P_t = \underset{s\rightarrow+\infty}{\lim}\mathbb{E}\big[V^a_s - V^b_s|\mathcal{F}_t\big],
	\end{equation}
	where $V^a$ (resp. $V^b$) is the cumulated volume of buy (resp. sell) market orders since the initial time $0$ and $(\mathcal{F}_t)_{t\geq 0}$ corresponds to the filtration generated by the order flow process. Hence the price moves when orders arrive on the market since market participants revise their anticipation about the long term cumulative imbalance of the order flow.\\
	
As for the transient part of the market impact, empirical measurements show that provided the execution rate of the metaorder is relatively constant, the function $MI$ is close to a power-law with respect to time, that is $MI(t)\sim t^{1-\alpha}$ with $\alpha\in (0,1)$, see \cite{BLL15, B10, FLM03, BDKLLT11}. More precisely, the coefficient $\alpha$ is found to be about $1/2$ so that the so-called square root law is approximately satisfied. Actually, it is proved in \cite{PRS17} that under some leverage neutrality assumption, the square root law can be simply derived from dimensional analysis.\\
	
We show in this work that under no-arbitrage assumption (represented by the linear permanent impact, the martingale price and thus \eqref{eq:price_process_def}), the market impact function has indeed to be a power-law of the form $MI(t) \sim t^{1-\alpha}$. Then we prove that for any $\alpha\in (0,1)$, the scaling limit of the price \eqref{eq:price_process_def} exists and satisfies
	\begin{equation*}
	\widehat{P}_t = B_{X_t}
	\end{equation*}
	with
	\begin{equation}
	\label{eq:intro:price_var_eq}
		X_t = \frac{2}{\delta}\int_{0}^{t}F^{\alpha, \lambda}(s) \mathrm{d}s + \frac{1}{\delta\sqrt{ \lambda}}\int_{0}^{t}F^{\alpha, \lambda}(t-s)\mathrm{d}W_{X_s },
	\end{equation}
where $W$ and $B$ are two Brownian motions, $\delta$ and $\lambda$ two positive constants and $F^{\alpha, \lambda}$ is the Mittag-Leffler cumulative distribution function, see Appendix \ref{annex:mittag} for definition. The correlation between the Brownian motions $B$ and $W$ is stochastic and related to the order flow imbalance. The above equation is a generalization of the rough Heston model \eqref{eq:intro_rough_vol}. Indeed we can show that when $ \alpha>1/2$, after differentiation, Equation \eqref{eq:intro:price_var_eq} can be rewritten under the form of \eqref{eq:intro_rough_vol} (up to a stochastic correlation factor)
with associated Hurst parameter $H = \alpha-1/2$. For $\alpha\leq 1/2$, we prove that $X$ is not continuously differentiable but has H\"older regularity $2\alpha-\varepsilon$ for any $\varepsilon>0$. Therefore we give to \eqref{eq:intro:price_var_eq} the name {\it hyper-rough Heston model} when $\alpha\leq 1/2$. Hence we are able to define rough Heston models for Hurst parameter in $(-\frac{1}{2}, \frac{1}{2}]$.\\
	
To obtain our results, our only modeling assumption is a dynamic for the order flow. More precisely, we consider for buy and sell market order arrivals two independent Hawkes processes and assume that each order is of unit size, see \cite{EEFR16,JR16b}. Recall that a Hawkes process $N$ is a self-exciting point process whose intensity $(\lambda_t)_{t\geq 0}$ is defined by
	\begin{equation*}
	\lambda_t = \mu + \int_{0}^{t}\phi(t-s)\mathrm{d}N_s,
	\end{equation*}
	with $\mu$ a positive constant and $\phi$ a non-negative locally integrable function. Such dynamic is a generalization of the Poisson process which is usually considered when modeling order flows, see among others \cite{Cont2010,cont2010stochastic,smith2003statistical}. It is non-parametric and very flexible so that it is really reasonable to assume that the actual order flow can be well approximated by a Hawkes based model. Note that we will not put any restriction
on the Hawkes parameters $\mu$ and $\phi$, except that they are similar for the the buy and sell flows. In this case, it is shown in \cite{J15} that the price process \eqref{eq:price_process_def} satisfies
	\begin{equation}
	\label{eq:price_kernel}
	P_t = P_0 + \int_{0}^{t}\xi(t-s)\mathrm{d}(N^a-N^b)_s,	\end{equation} with
	\begin{equation}	
	\label{eq:xi}\xi(t) =  1+ \Big(1+ \int_0^{+\infty} \psi(u)\mathrm{d}u\Big)\int_{t}^{+\infty}  \phi(u)\mathrm{d}u
	\end{equation}
	and
$$\psi = \sum_{i\geq 1}(\phi)^{*i},
$$ where $(\phi)^{*1}=\phi$ and for $k\geq 2$, $(\phi)^{*k}$ denotes the convolution product of $(\phi)^{*(k-1)}$ with $\phi$.\\

Using a rescaling procedure to describe the macroscopic behavior of \eqref{eq:price_kernel}, we show that only one very subtle specification of the Hawkes processes can lead to a non-trivial market impact, which has to be power-law. Furthermore, it implies that the market is highly endogenous. In addition, depending on the market impact shape, the scaling limit of the price is a rough or hyper-rough Heston model \eqref{eq:intro:price_var_eq}, with a one-to-one correspondence between the exponent of the impact function and the Hurst parameter of the volatility.\\
	
	The paper is organized as follows. In Section \ref{sec:market_impact}, we show that under the assumption that the market impact function is not degenerate, it can only be a power-law with parameter $1-\alpha$ for some $\alpha\in(0,1)$. Then in Section \ref{sec:scaling_limits} we explain that the macroscopic limit of (\ref{eq:price_kernel}) is a rough or hyper-rough Heston model with Hurst parameter $H=\alpha-1/2$.

\section{Market impact is power-law}
\label{sec:market_impact}
In this section, we show that if there exists a non-degenerate market impact function, it has to be a power-law. Moreover we will see that it implies a highly endogenous market. By non-degenerate we essentially mean a market impact function which is ultimately decreasing for buy metaorders (and conversely for sell metaorders), see Assumption \ref{assumption:market_impact}. This is the formalization of the two phases behavior of market impact discussed in the introduction.

\subsection{Asymptotic framework and metaorders modeling}
\label{subsec:scaling_metaorder_parameters}
Let $T$ be our final horizon time for the metaorders we will define in the sequel. Recall that the market order flow on $[0,T]$ (and after $T$) is given by two Hawkes processes with same parameters, $N^a$ for the buy market orders and $N^b$ for the sell orders. Since the time-length of a metaorder is typically large compared to the inter-arrivals of individual market orders, it is natural to consider that $T$ goes to infinity.\\

We want to work in a general setting which enables us to be potentially compatible with empirical studies showing that markets are highly endogenous. In the Hawkes process context, the degree of endogeneity of the market is measured by the $L^1$ norm of $\phi$, denoted by $\|\phi\|_1$, see \cite{FS15,BBH13,JR16b,jaisson2015limit}. Therefore a highly endogenous market corresponds to the case where $\|\phi\|_1$ is close but smaller than unity. So we allow the model parameters to possibly depend on $T$. Thus, from now on, we use the superscript $T$ for all quantities that could depend on $T$. In particular $\|\phi^T \|_1$ may go to one as $T$ tends to infinity. We also write  $N^{a,T}$, $N^{b,T}$, $\mu^T$, $\phi^T$ to describe the market order flow and model parameters corresponding to the time-horizon $T$, and we set $\phi^T = a^T\phi $ for $\phi$ a non-negative function such that $\|\phi\|_1=1$ and $(a^T)_{T\geq0}$ a real sequence in $(0,1)$. Note that we do not impose that $!
 a^T$ goes
  to one. In fact we will show that it is a necessary condition for the existence of a non-degenerate market impact function.\\
	
We finally need to define a formalism for a sequence of buy (say) metaorders which will be added to the global order flow. We assume that a metaorder is split through unitary market orders over $[0,T]$. In the spirit of \cite{J15}, we consider that the arrival times of the market orders are given by a non-homogenous Poisson process with intensity
	\begin{equation*}
	\nu^T(t)=I^Tf(\frac{t}{T}),
	\end{equation*}
with $I^T$ a sequence of non-negative real numbers and $f$ a non-negative continuous function on $[0,1]$ with positive integral. Hence the metaorder average size is $ I^T T\|f\|_1$ and the order of magnitude of its duration is essentially $T$. Note that we allow $f$ to be different from a constant to get more realistic splitting schemes than those given by constant rate Poisson processes, see for example \cite{AC01}.\\
	
To compute the market impact function in practice, one typically considers the empirical mean of the price movements over many metaorders with similar durations and volumes counted in proportion of the total traded volume. So in our setting, it is natural to take $I^T\times T$ essentially proportional to the total number of other orders executed over $[0,T]$. Thus we consider
	\begin{equation*}
	I^T = \gamma \beta^T,\,\,\mathrm{with}\,\, \beta^T = \mu^T(1-a^T)^{-1},	
	\end{equation*}
where $\gamma<1$ and $\beta^T$ is the average intensity of a stationary version\footnote{Rigorously speaking, our Hawkes processes are not stationary since they start at time $t=0$ and not $t=-\infty$.} of the Hawkes process $N^{a,T}$. Thus, the proportion of the order flow which is due to the considered metaorder is essentially $\gamma/(1+\gamma)$ and $\gamma$ will be considered reasonably small.

\subsection{Market impact in the Hawkes setting}
\label{subsec:market_impact_hawkes}
In this section the parameter $T$ is fixed.  Assuming that the volume of our metaorder is small enough, the total order flow is not deeply modified by it. Hence other agents do not observe significant changes in the order flow dynamics. So the way the market reacts to the incoming orders remains unchanged. Recall that in our model, the market reaction to the order flow (without our metaorder) is given by \eqref{eq:price_kernel}.\\

We work under the setting of the previous section assuming that the number of assets bought through our metaorder is a non-homogenous Poisson process $(n^T_{t})_{t\geq 0}$. Therefore we obtain
	\begin{equation*}
	    P^T_t = P_0 +  \int_{0}^{t}\xi^T(t-s)\mathrm{d}(N^{a,T}-N^{b,T}+n^T)_s,
	\end{equation*}
where $(N^{a,T}_t, N^{b,T}_t)_{t\geq 0}$ corresponds to the aggregated order flows of all other agents. Thus the market digests the order flow as if it is a bivariate Hawkes process with parameters $\mu$ and $\phi$ ($\gamma$ being small).\\
	
Now we are in the position to properly compute the market impact function of our metaorder. We have
\begin{equation*}
	MI^T(t) = \mathbb{E}[P^T_t-P_0] = \int_{0}^{t}\xi^T(t-s)\mathbb{E}[\mathrm{d}n^T_s].
\end{equation*}
This equation together with \eqref{eq:xi} shows that for any $t\geq 0$, the market impact function can be decomposed into two terms as follows:
\begin{equation*}
    MI^T(t) = PMI^T(t) + TMI^T(t),
\end{equation*}
where
$$
PMI^T(t) = \mathbb{E}[n^T_t]
$$ 
and
$$
TMI^T(t) = \int_0^t \Gamma^T(t-s)\mathbb{E}[\mathrm{d}n^T_s],
$$
with
\begin{equation*}
    \Gamma^T(s) = (1 - a^T)^{-1} \int_{s}^{+\infty}  \phi^T(u)\mathrm{d}u,
\end{equation*}
where we have used the fact that
$$
\int_0^{+\infty}\psi^T(u) \mathrm{d}u =\sum_{k\geq 1}^{+\infty} \big(\int_0^{+\infty}\phi^T(u)\mathrm{d}u\big)^k= \frac{a^T}{1-a^T}.
$$
Note that the definition of $PMI^T(t)$ is compatible with that of $PMI$ given in the introduction. Indeed since the order intensity from our metaorder is eventually null and $\Gamma^T(t)$ tends to zero as $t$ goes to infinity,  we get  
\begin{equation*}
    \underset{t\rightarrow +\infty}{\lim}TMI^T(t) =0.
\end{equation*}
The effect of the term $TMI^T$ is thus only temporary. That is why this term is called transient part of the market impact. 

\subsection{Scaling limit of the market impact}
\label{subsec:scaling_limit}
We now rescale the market impact function as the horizon time $T$ goes to infinity. If the sequence of rescaled market impact functions converges, we call its limit macroscopic market impact function.\\

First we reparametrize in time and consider $(MI^T(f,tT))_{t\in\mathbb{R}^+}$ (we put the function $f$ as parameter of $MI^T$ to insist on the fact that the market impact function depends on the metaorder strategy). Thus $t=1$ corresponds essentially to the end of the metaorder. Regarding the scaling in space, since in our framework the size of a metaorder is measured relatively to the total volume, which is of order $T\beta^T$ on $[0,T]$, we finally define our rescaled market impact function $\overline{MI}^T$ on $\mathbb{R}^+$ by
\begin{equation*}
	\overline{MI}^T(f,t) = \frac{MI^T(f,tT)}{T \beta^T}= \overline{PMI}^T(f,t) + \overline{TMI}^T(f,t),
	\end{equation*}
	with
	$$\overline{PMI}^T(f,t) = \gamma\int_0^t f(x)\mathrm{d}x$$
	and
	$$\overline{TMI}^T(f,t) = \gamma \frac{a^T(1-a^T)^{-1}}{T}\int_0^{Tt}f(t-x/T) \int_{x}^{+\infty}\phi(u)\mathrm{d}u\mathrm{d}x.
$$
Remark that the permanent impact term does not depend on $T$. Thus there always exists a macroscopic permanent market impact function and the convergence of the sequence $(\overline{MI}^T(f,\cdot))_{T\geq0}$ is equivalent to that of $(\overline{TMI}^T(f,\cdot))_{T\geq0}$. Motivated by the empirical results on market impact \cite{BLL15,B10,GW15,FLM03,BP03} discussed in the introduction, we make the following natural assumption.
	\begin{assumption}
		\label{assumption:market_impact}
		For constant execution rate, that is $f = \mathbf{1}_{[0,s]}$ for some $s \in (0,1]$, the scaling limit of the market impact function exists pointwise and is non-increasing after time $s$. Furthermore, there exists $t>s$ such that the value of this limiting function at time $t$ is smaller than that at time $s$. 
	\end{assumption}
We will see that under some sets of parameters, Assumption \ref{assumption:market_impact} is indeed satisfied in our model. It implies that for $f = \mathbf{1}_{[0,s]}$ with $s\in (0,1]$, we can define the macroscopic market impact function $\widehat{MI}(f,t)$ and its transient and permanent components $\widehat{TMI}(f,t)$ and $\widehat{PMI}(f,t)$ as
	\begin{equation*}
	\widehat{MI}(f,t) = \underset{T\rightarrow+\infty}{\lim}\overline{MI}^T(f,t),~~\widehat{TMI}(f,t) = \underset{T\rightarrow+\infty}{\lim}\overline{TMI}^T(f,t),~~\widehat{PMI}(f,t) = \underset{T\rightarrow+\infty}{\lim}\overline{PMI}^T(f,t).
	\end{equation*}
    Using Tauberian theorems, see Appendix \ref{annex:karamata}, we obtain the following result.
	\begin{theorem}
		\label{th:market_impact_limit}
		Under Assumption \ref{assumption:market_impact}, for any measurable non-negative function $f$ defined on $\mathbb{R}^+$, continuous on $[0,1]$ and supported on $[0,1]$, the macroscopic market impact function and its transient part exist. More precisely, there exists a parameter $\alpha\in (0, 1]$ such that for any $t>0$, when $\alpha<1$,
		\begin{equation}
		\label{eq:tail_kernel_a}
	    \underset{T \rightarrow +\infty}{\lim}\overline{TMI}^T(f,t) =\gamma K (1-\alpha)  \int_0^tf(t-u)u^{-\alpha} \mathrm{d}u,
		\end{equation}
		for some $K>0$, and when $\alpha=1$
		\begin{equation}
		\label{eq:tail_kernel_b}
		\underset{T \rightarrow +\infty}{\lim}\overline{TMI}^T(f,t) = \gamma Kf(t).
		\end{equation}
		Furthermore, the Hawkes kernel $\phi$ necessarily satisfies
		\begin{equation*}
		\label{eq:structu_a}
		\int_0^t \int_{s}^{+\infty}\phi(u)\mathrm{d}u\mathrm{d}s  = t^{1-\alpha } \,L(t),
		\end{equation*}
		where $L$ is a slowly varying function (see definition in Appendix \ref{annex:karamata}).
		Finally we necessarily have
		\begin{equation*}
		\label{eq:structu_b}
			(1-a^T)^{-1}T^{-\alpha}L(T){\rightarrow}K,
		\end{equation*}
		and consequently $a^T\rightarrow 1$ (see Proposition \ref{prop:annex_slowly_varying} in Appendix \ref{annex:karamata}).
	\end{theorem}
Considering for example $f = \mathbf{1}_{[0,1]}$, Theorem \ref{th:market_impact_limit} shows that under no-arbitrage together with the assumption of the existence of the macroscopic market impact function, the transient part of the market impact is power-law while the permanent part is linear. Moreover Equation \eqref{eq:tail_kernel_a} gives that the decay of the market impact is essentially a power-law with exponent $-\alpha$, see Figure \ref{fig:market_impact_d_5} for illustration.\\

	\begin{figure}[tbph]
	\centering
	\label{fig:market_impact_d_5}
	\includegraphics[width = \figurewidth, height = \figureheight]{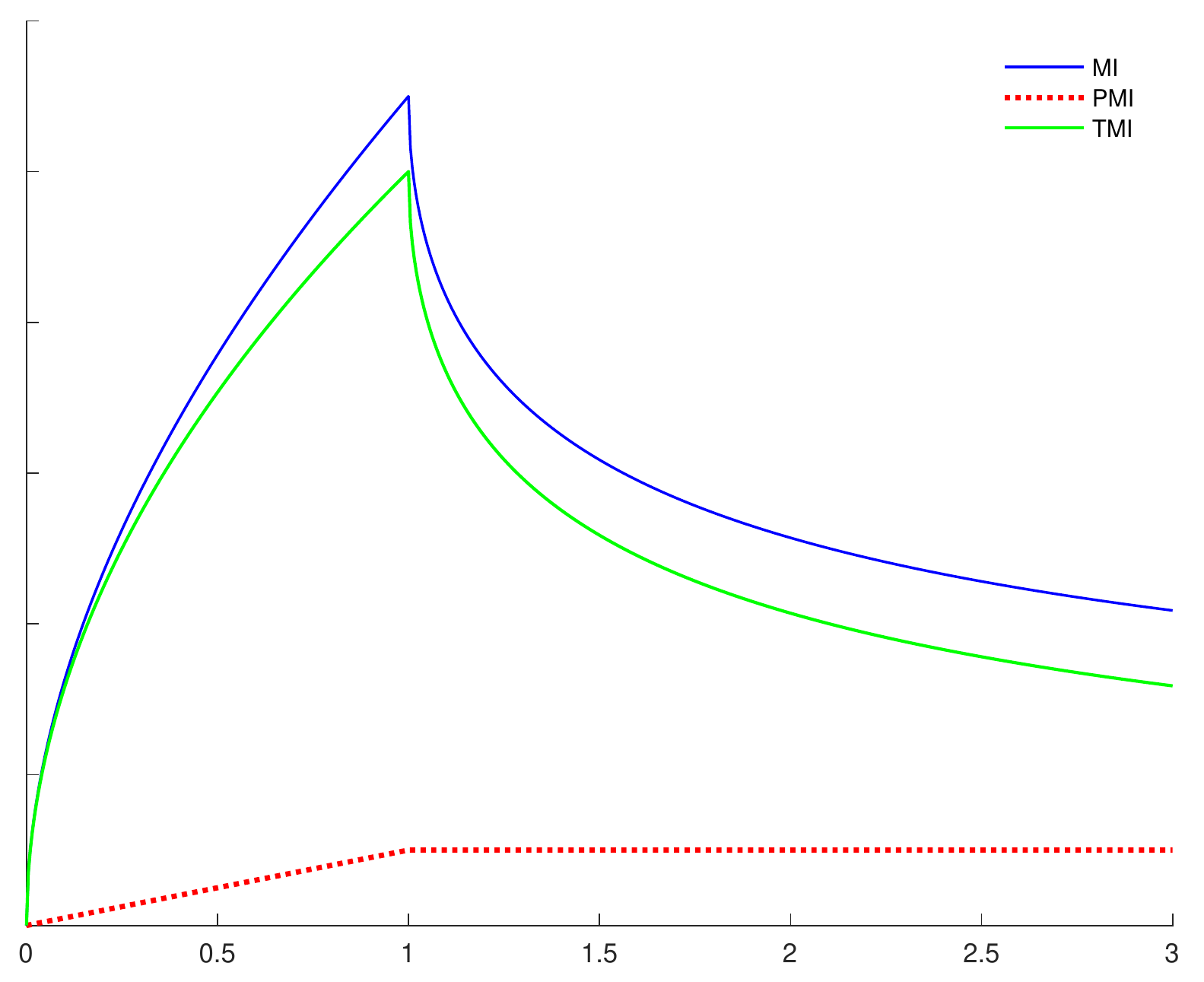}
	\caption{Illustration of the decomposition of the macroscopic market impact function for a metaorder executed uniformly over $[0,1]$, with $\alpha = 0.5$. Time is on the $x-$axis.}
\end{figure} 
	
The fact that $a_T$ goes to one implies that the non-linear transient part of the market impact (case $\alpha<1$) can arise only in a highly endogenous market. This non-linearity means that the market reacts differently to child market orders of a metaorder depending on their position within the metaorder. This is possible only if the time range of the persistence of the global order flow is of similar magnitude as the typical length of a metaorder, which is large compared to the inter-arrivals of market orders. Using the population approach to Hawkes processes, see \cite{BM14a,FS15,jaisson2015limit,JR16b}, it is easily seen that in our model such property can hold only provided $a_T$ goes to one.\\
	
In this regard the case $\alpha = 1$ is quite degenerate since the market has somehow no memory and reacts the same way to market orders, independently of their position within the metaorder. Even more, the price instantaneously decreases to its permanent level when the metaorder is completed. This means that the market is able to detect instantaneously the end of a metaorder, which seems unrealistic and incompatible with empirical measurements.

\section{Macroscopic limit of the price}
	\label{sec:scaling_limits}
We finally show in this section that under Assumption \ref{assumption:market_impact}, the macroscopic price, that is the limit as time goes to infinity of the properly rescaled microscopic price (\ref{eq:price_kernel}), is diffusive with rough or hyper-rough volatility. Moreover, we explicit the link between the market impact shape exponent and the Hurst parameter of the volatility. 
	
	\subsection{Scaling limit of the price process}
	We start with an assumption which is necessary to get a non-trivial long term limit for the price (\ref{eq:price_kernel}).
	\begin{assumption}
		\label{assumption:structure}
	For some $\delta>0$, we have 
		\begin{equation*}
		\label{assumption:mu}
		(1-a^T)\mu^T T\underset{T\rightarrow+\infty}{\rightarrow} \delta.
		\end{equation*}
	\end{assumption}
	Assumption \ref{assumption:structure} is classical in the context of Hawkes processes with kernel whose $L^1$ norm tends to one, see \cite{JR16b}. Indeed, it ensures that the number of events does not explode asymptotically.\\
	
	According to Equation \eqref{eq:price_process_def}, price and volume are homogenous. Therefore we rescale the price the same way as the metaorders. Taking for simplicity and without loss of generality $P_0=0$, we define for $t\in[0,1]$
\begin{equation*}
	\overline{P}_t^T = \frac{1}{T\beta^T} P^T_{tT}=\frac{1-a^T}{T\mu_T} \int_{0}^{t}\xi^T\big(T(t-s)\big)\mathrm{d}\big(N^{a,T}_{tT} - N^{b,T}_{tT}\big),
	\end{equation*}
	where
	\begin{equation*}
	\xi^T(t) = \big(1 + \int_0^{+\infty} \psi^T(u)\mathrm{d}u\big)\big(1 - \int_{0}^{t}\phi^T(u)\mathrm{d}u\big).
	\end{equation*}
Let $\alpha$ be the parameter of the market impact function in Theorem \ref{th:market_impact_limit}, $K$ the constant introduced in Equations \eqref{eq:tail_kernel_a} and \eqref{eq:tail_kernel_b} and $\lambda = (K\Gamma(2-\alpha))^{-1}$. Let $B^a$ and $B^b$ be two independent Brownian motions and $X^a$ and $X^b$ be defined by
\begin{equation*}
X^a_t = \int_{0}^{t}F^{\alpha, \lambda}(s) \mathrm{d}s + \frac{1}{\sqrt{\delta\lambda}}\int_{0}^{t}F^{\alpha, \lambda}(t-s)\mathrm{d}B^a_{X^a_s}
\end{equation*}
and $X^b$ is solution of the same equation replacing the superscript $a$ by $b$. We have the following result for the macroscopic limit of the price process, whose proof is given in Section \ref{proof:th:scaling_limit_price}.

\begin{theorem}
		\label{th:scaling_limit_price}
		Under Assumptions \ref{assumption:market_impact} and \ref{assumption:structure}, the sequence of rescaled price processes $(\overline{P}^T)_{T\geq0}$ converges in law for the Skorokhod topology towards a process $\widehat{P}$ such that for $t\in [0,1]$
		\begin{equation*}
		\widehat{P}_t =  \frac{1}{\sqrt{\delta}}\big(B^a_{X^a_t} - B^b_{X^b_t}\big).
		\end{equation*}
		In particular, there exists a Brownian motion $W$ such that the integrated variance $X = \big( X^a + X^b\big) / \delta$ of $\widehat{P}$ is solution of the stochastic rough Volterra equation
		\begin{equation}
		\label{eq:price_var_eq}
		X_t = \frac{2}{\delta}\int_{0}^{t}F^{\alpha, \lambda}(s) \mathrm{d}s + \frac{1}{\delta\sqrt{\lambda}}\int_{0}^{t}F^{\alpha, \lambda}(t-s)\mathrm{d}W_{X_s}.
		\end{equation}
		Moreover, for any $\varepsilon>0$, the process $X$ has H\"older regularity $1\wedge(2\alpha - \varepsilon)$. It is continuously differentiable for $\alpha>1/2$ and not continuously differentiable for $\alpha\leq 1/2$.
\end{theorem}

Theorem \ref{th:scaling_limit_price} shows that the no-arbitrage principle together with the existence of market impact imply that the macroscopic price is a diffusive process whose cumulative variance is solution of a stochastic rough Volterra equation (except when $\alpha =1$ which corresponds to the classical Heston model, see Corollary \ref{th:price_vol}). Note that $X$ plays the role of an integrated variance and that when $\alpha\leq 1/2$ it is not continuously differentiable. Thus, in that case, the spot variance is not well-defined and only its integrated version makes sense. This is why for $\alpha\leq 1/2$, we call this model hyper-rough volatility model (more precisely hyper-rough Heston model, see below).\\

	From Theorem 3.2 in \cite{JR16b}, we have that for $\alpha>1/2$, the process $X^a$ is almost surely differentiable and its derivative $Y^a$ is the unique solution of 
	\begin{equation*}
		Y^a_t = \frac{\lambda}{\Gamma(\alpha)}\big(\int_{0}^{t}(t-s)^{\alpha-1} (1- Y^a_s) \mathrm{d}s + \frac{1}{\sqrt{\delta\lambda}}\int_{0}^{t}(t-s)^{\alpha-1}\sqrt{Y^a_s}\mathrm{d}B^a_s\big).
	\end{equation*}
The same result holds for $Y^b$ replacing the superscript $a$ by $b$. We deduce that when $\alpha>1/2$, the integrated volatility admits a derivative and the macroscopic limit of the price follows a rough Heston model. More precisely, we have the following corollary.
	\begin{corollary}
		\label{th:price_vol}
		When $\alpha>1/2$, the process $X$ is differentiable almost surely and its derivative $Y$ is the unique solution of the stochastic rough Volterra equation
		\begin{equation*}
		Y_t = (Y_t^a+Y_t^b)/\delta=\frac{\lambda}{\Gamma(\alpha)}\big(\int_{0}^{t}(t-s)^{\alpha-1} (\frac{2}{\delta}-Y_s) \mathrm{d}s + \frac{1}{\delta\sqrt{\lambda}}\int_{0}^{t}(t-s)^{\alpha-1}\sqrt{Y_s}\mathrm{d}W_s\big),
		\end{equation*}
with $W$ a Brownian motion. Furthermore the dynamic of the price $\widehat{P}$ is
		\begin{equation*}
		    \mathrm{d}\widehat{P}_t =\frac{1}{\sqrt{\delta}}\big( \sqrt{Y^a_t}\mathrm{d}B^a_t - \sqrt{Y^b_t}\mathrm{d}B^b_t \big).
		\end{equation*}
	\end{corollary}
This result highlights the fact that at the macroscopic limit, the correlation $\rho_t$ between the two Brownian motions driving price and volatility is stochastic. More precisely we have
\begin{equation*}
    \rho_t= \frac{Y^a_t - Y_t^b}{Y^a_t + Y^b_t}.
\end{equation*}
Hence the correlation sign depends on that of $Y^a_t - Y_t^b$. The process $Y^a$ (resp. $Y^b$) corresponding to the volatility of the ask (resp. bid) side of the market (see Step 4 in Section \ref{proof:th:scaling_limit_price}), this can be interpreted in terms of order flow dynamics. Indeed suppose that $Y^a\gg Y^b$ and that price is increasing. Then the instantaneous imbalance has the same sign as price returns. Thus the volatility increases as the order flow excites the price dynamic. Conversely, if the price increases and $Y^a \ll Y^b$, the volatility decreases since the order flow tends to compensate the upward price variation.\\

To prove the convergence in law in Theorem \ref{th:scaling_limit_price}, we show that $(\overline{P}^T)_{T\geq 0}$ is tight and that all limit points have the same law. This is done using the characteristic function of Hawkes processes in the spirit of \cite{EER16}. 
A direct proof would consist in obtaining uniqueness in law for solutions of Equation \eqref{eq:price_var_eq} as done in \cite{JLP17} for $\alpha>1/2$. However, such approach seems quite intricate to adapt for $\alpha\leq \frac{1}{2}$. We have the following result whose proof is given in Section \ref{proof:th:characteristic_function_variance}.
	\begin{theorem}
		\label{th:characteristic_function_variance}
		Let $X$ be the cumulated variance process given in Theorem \ref{th:scaling_limit_price} and $h$ a continuously differentiable function from $\mathbb{R}^{+}$ to $\mathbb{R}$ such that $h(0)=0$. The function
		\begin{equation*}
		K(h, t) = \mathbb{E}\big[\emph{exp}\big(\int_0^t i h(t-s)\mathrm{d}X_{s}\big)\big]
		\end{equation*}
		satisfies
		\begin{equation*}
		K(h,t) = \emph{exp}\big( \int_{0}^{t}g(s) \mathrm{d}s \big),
		\end{equation*}
		with $g$ the unique continuous solution of the Volterra Riccati equation
		\begin{equation}
		\label{eq:volterra_riccati}
		g(t) = \int_0^t f^{\alpha, \lambda}(t-s)\big(\delta^{-1}\frac{1}{ 4}g(s)^2 + \delta^{-1}2ih(s)\big)\mathrm{d}s,
		\end{equation}
		where $f^{\alpha, \lambda}$ is the Mittag-Leffler density function, see Appendix \ref{annex:mittag}. 
	\end{theorem} 
Theorem \ref{th:characteristic_function_variance} extends some already known results about characteristic functions related to rough Heston models for $\alpha > \frac{1}{2}$, see \cite{EGR18, EER16,JLP17}. Note that the characteristic function of the macroscopic price process $\widehat{P}_t$ can also be obtained using the same type of proof as that for Theorem \ref{th:characteristic_function_variance}.
	
\subsection{Conclusion}
Using only the no-arbitrage principle and the assumption that market impact exists and has a transient component, we have shown in a general framework for the order flow that the market impact function can only be a power-law with exponent $1-\alpha$ for some $\alpha\in(0,1)$ (we drop here the case $\alpha=1$ which leads to a somehow degenerate market impact function). The parameter $\alpha$ also appears necessarily in the tail of the kernel of the Hawkes process driving the order flow: $\phi(x)\sim x^{-(1+\alpha)}$ as $x$ goes to infinity. Furthermore, this also implies that the market is highly endogenous. Even more interestingly, we obtain that the macroscopic behavior of the price is that of a rough or hyper-rough Heston model with Hurst parameter $H=\alpha-1/2$.\\

The relationship between market impact, tail of Hawkes kernel and volatility Hurst parameter allows us to confront our results to empirical measurements. In \cite{BLL15} it is found that the market impact function fits a power-law with exponent $0.45$. In \cite{GJR14} it is shown that volatility is rough with a Hurst parameter of order $0.1$. Finally in \cite{BBH13}, the authors calibrate a Hawkes process on market orders arrival and obtain that the kernel decays as a power-law function with exponent around $-1.45$. All these measurements are compatible with our results (and suggest that market impact is close to square root).
	
	\section*{Acknowledgments}
	We thank Jean-Philippe Bouchaud, Omar El Euch, Masaaki Fukasawa and Jim Gatheral for many interesting discussions. The authors gratefully acknowledge the financial support of the {\it ERC Grant 679836 Staqamof} and of the chair {\it Analytics and Models for Regulation}.
	
	\section{Proofs}
	\label{sec:proofs}
	
	\subsection{Proof of Theorem \ref{th:market_impact_limit}}
	\label{proof:th:market_impact_limit}
	Let $f = \mathbf{1}_{[0,s]}$, $s\in (0,1]$. From Assumption \ref{assumption:market_impact}, we have the pointwise convergence of $(\overline{MI}^T(f,\cdot))_{T\geq0}$. As previously explained, this is equivalent to the convergence of $(\overline{TMI}^T(f,\cdot))_{T\geq0}$. Moreover, $(\overline{PMI}^T(f,\cdot))_{T\geq0}$ being independent of $T$, Assumption \ref{assumption:market_impact} implies that the sequence of functions
	\begin{equation}
	\label{eq:impact_a}
		\overline{TMI}^T(f,t) =\gamma\int_0^{t}a^T(1-a^T)^{-1}\int_{yT}^{+\infty}\phi(u)\mathrm{d}u f(t-y)\mathrm{d}y,
	\end{equation}
	converges pointwise. The function $\phi$ being non-negative and integrable, $\overline{TMI}^T(f,\cdot)$ is non-negative, non-decreasing and concave on $[0,s]$ and then non-increasing. Hence $\overline{TMI}^T(f,\cdot)$ reaches its maximum in $s$. By pointwise convergence, $\widehat{TMI}(f, \cdot)$ has the same properties. Since we have assumed that $\widehat{MI}(f,t)<\widehat{MI}(f,s)$ for some $t>s$ and $\widehat{PMI}(f, \cdot)$ is non-decreasing, we deduce that $\widehat{TMI}(f,s)>0$.\\
	
	Let $g(t) = \gamma^{-1}\widehat{TMI}(\mathbf{1}_{[0,t]},t)$ for $t\in (0,1]$ and consider
	\begin{equation*}
	    R(t) = \int_0^t \int_y^{+\infty}\phi(u)\mathrm{d}u\mathrm{d}y>0.
	\end{equation*}
	According to Equation \eqref{eq:impact_a}, we have for $t\in (0,1]$
	\begin{equation}
	\label{eq:impact_c}
		\frac{R(Tt)}{R(T)}  \underset{T\rightarrow+\infty}{\rightarrow}\frac{g(t)}{g(1)}>0.
	\end{equation}
	By the characterisation theorem, see Theorem \ref{th:characterisation} in Appendix \ref{annex:karamata}, we deduce that the previous limit holds for all $t>0$ with some suitable extension of the function $g$. Moreover there exist some $\beta\in \mathbb{R}$, $K>0$ and $ L$ a slowly varying function such that for $t>0$
	\begin{equation*}
		g(t) = Kt^{\beta},\,\, R(t) = L(t)t^{\beta}.
	\end{equation*}
	Remark that for $t\in (0,1]$, we have $g(t) = \widehat{TMI}(\mathbf{1}_{[0,1]},t)$, which is concave. Thus $\beta\in [0,1]$.	Taking $s = t = 1$ in the pointwise convergence \eqref{eq:impact_a}, we get
	\begin{equation}
	\label{eq:impact_b}
	 a^T\frac{(1-a^T)^{-1}}{T}\int_0^{T}\int_{y}^{+\infty}\phi(u)\mathrm{d}u\mathrm{d}y  =  a^T(1-a^T)^{-1} T^{\beta-1}L(T) \underset{T\rightarrow +\infty}{\rightarrow} g(1) = K>0.
	\end{equation}
	Consider now the sequence of functions 
	\begin{equation*}
	    \overline{\Gamma}^T(y) = a^T(1-a^T)^{-1}\int_{Ty}^{+\infty}\phi(u)\mathrm{d}u.
	\end{equation*}
	We get from \eqref{eq:impact_c}, \eqref{eq:impact_b} and property of slowly varying function that for  any $t>0$
	\begin{equation*}
	\underset{T\rightarrow +\infty}{\lim} \int_0^{t} \overline{\Gamma}^T(y)\mathrm{d}y = K t^{\beta}.
	\end{equation*}
	Suppose that $\beta\neq 0$. Let $0\leq a< b$ and $s\in [a,b]$. We have
	\begin{equation*}
	\underset{T\rightarrow +\infty}{\lim} \int_a^{s}\frac{ \overline{\Gamma}^T(u)}{\int_a^{b} \overline{\Gamma}^T(v)\mathrm{d}v}\mathrm{d}u  = \frac{s^{\beta}-a^{\beta}}{b^{\beta}-a^{\beta}}.
	\end{equation*}
	The right hand side is the cumulative distribution function of a random variable with support on $[a,b]$ whose law is denoted by $m_{a,b}^{\beta}$. Hence we have the convergence in law
	\begin{equation*}
	    \mathbf{1}_{[a,b]}\frac{ \overline{\Gamma}^T(u)\mathrm{d}u}{\int_a^{b} \overline{\Gamma}^T(v)\mathrm{d}v} \underset{T\rightarrow +\infty}{\rightarrow}m_{a,b}^{\beta}(\mathrm{d}u).
	\end{equation*}
	So for any bounded continuous function $g$ on $[a,b]$, we get
	\begin{equation*}
	\underset{T\rightarrow +\infty}{\lim} \int_a^{b}\frac{ \overline{\Gamma}^T(u)}{\int_a^{b} \overline{\Gamma}^T(v)\mathrm{d}v}g(u)\mathrm{d}u  = \int_a^b g(u)m_{a,b}^{\beta}(\mathrm{d}u).
	\end{equation*}
	Consequently,
	\begin{equation*}
	\underset{T\rightarrow +\infty}{\lim} \int_a^{b}\overline{\Gamma}^T(u)g(u)\mathrm{d}u  =  K \int_a^b g(u)m_{a,b}^{\beta}(\mathrm{d}u)\big(b^{\beta}-a^{\beta} \big) = K \beta \int_a^b g(u)u^{\beta-1}\mathrm{d}u .
	\end{equation*}
Now let $f$ be a non-negative measurable function defined on $\mathbb{R}^+$, continuous on $[0,1]$ and supported on $[0,1]$. For $t\leq 1$ we have
	\begin{equation*}
	    \int_0^t f(t-u) \overline{\Gamma}^T(u)\mathrm{d}u \underset{T\rightarrow +\infty}{\rightarrow} K \beta \int_0^t f(t-u) u^{\beta-1}\mathrm{d}u
	\end{equation*}
	and for $t> 1$
	\begin{equation*}
	    \int_0^t f(t-u) \overline{\Gamma}^T(u)\mathrm{d}u = \int_{t-1}^t f(t-u) \overline{\Gamma}^T(u)\mathrm{d}u\underset{T\rightarrow +\infty}{\rightarrow} K \beta \int_{t-1}^t f(t-u) u^{\beta-1}\mathrm{d}u.
	\end{equation*}
	Finally for any $t\geq 0$
	\begin{equation*}
	    \widehat{TMI}(f, t) = \underset{T\rightarrow+\infty}{\lim}\overline{TMI}^T(f, t) =  \gamma K\beta \int_0^t f(t-u)u^{\beta-1}\mathrm{d}u.
	\end{equation*}
	Thus when $\beta\in (0,1]$, we have the existence of a macroscopic limit for the transient part of the market impact function (and therefore for the market impact function). 
	Remark that for $\beta  = 1$
    \begin{equation*}
	\widehat{TMI}(f,t) =\gamma K\int_0^{t}f(u)\mathrm{d}u.
	\end{equation*}
	Consequently, in that case, $\widehat{TMI}(\mathbf{1}_{[0,1]},\cdot)$ is a non-decreasing function. This is in contradiction with Assumption \ref{assumption:market_impact}, hence $\beta$ cannot be equal to $1$.\\
	
	Suppose that $\beta = 0$. For any $t>0$ we have
	\begin{equation*}
	    \mathbf{1}_{[0,t]}\frac{ \overline{\Gamma}^T(u)\mathrm{d}u}{\int_0^{t} \overline{\Gamma}^T(v)\mathrm{d}v} \underset{T\rightarrow +\infty}{\rightarrow}\delta_0(\mathrm{d}u),
	\end{equation*}
	where $\delta_0$ is the Dirac measure in $0$. Then for any bounded continuous function $g$
	\begin{equation*}
	\underset{T\rightarrow +\infty}{\lim} \int_0^{t}\frac{ \overline{\Gamma}^T(u)}{\int_0^{t} \overline{\Gamma}^T(v)\mathrm{d}v}g(u)\mathrm{d}u  = g(0).
	\end{equation*}
Now let $f$ be a non-negative measurable function defined on $\mathbb{R}^+$, continuous on $[0,1]$ and supported on $[0,1]$. For $t\leq 1$ we have
	\begin{equation*}
	    \int_0^t f(t-s) \overline{\Gamma}^T(s)\mathrm{d}s \underset{T\rightarrow +\infty}{\rightarrow} K f(t)
	\end{equation*}
	and for $t>1$
	\begin{equation*}
	    0 \leq \int_0^t f(t-s) \overline{\Gamma}^T(s)\mathrm{d}s \leq \int_0^t \tilde{f}(t-s) \overline{\Gamma}^T(s)\mathrm{d}s \underset{T\rightarrow +\infty}{\rightarrow} 0,
	\end{equation*}
	with $\tilde{f}$ is a non-negative continuous extension of $f\mathbf{1}_{[0,1]}$ on $\mathbb{R}^+$ supported on $[0,1+\frac{t-1}{2}]$. Finally for any $t\geq 0$
    \begin{equation*}
	    \widehat{TMI}(f, t)=\gamma K f(t).
	\end{equation*}
	Consequently for $\beta=0$, we also have the existence of a macroscopic limit for the transient part of the market impact function (and therefore of the market impact function). We obtain the result letting $\alpha = 1-\beta$.\\
	
Finally note that there do exist some model parameters such that Assumption \ref{assumption:market_impact} is satisfied. For example any kernel $\phi$ such that $\phi(t)\sim_{+\infty} c t^{-\alpha-1}$ with $c>0$.
	
	\subsection{Proof of Theorem \ref{th:scaling_limit_price}}
	\label{proof:th:scaling_limit_price}
We proceed in five steps.
	\begin{enumerate}
			\item Step $1$: We first prove a preliminary result on the characteristic function of Hawkes processes that we use later in Step $3$.
		\item Step $2$: We rewrite the sequence $(\overline{P}^T)_{T\geq 0}$ in a convenient way.
		\item Step $3$: We adapt results from \cite{EER16} and \cite{JR16b} on scaling limits of nearly unstable heavy-tailed Hawkes processes to our more general framework.
		\item Step $4$: We deduce from the previous steps the convergence in law for the Skorokhod topology of the sequence $(\overline{P}^T)_{T\geq 0}$ and explicit the equation satisfied by the limit.
        \item Step $5$: We prove the results on the regularity of solutions of Equation \eqref{eq:price_var_eq}.
	\end{enumerate} 
For simplicity and without loss of generality we take $P_0 = 0$.

    \subsubsection{Step $1$}
    We derive a result on the characteristic function of Hawkes processes using similar arguments as those introduced in \cite{EER16}. Recall that the notation $*$ stands for the convolution product on $\mathbb{R}^+$. More precisely for $f$ and $g$ suitable measurable functions and $m$ a measure
	\begin{equation*}
	    (f*g)(t) = \int_0^t f(t-s)g(s) \mathrm{d}s
	\end{equation*}
	and 
    \begin{equation*}
	    (f*\mathrm{d}m)(t) = \int_0^t f(t-s)m(\mathrm{d}s). 
	\end{equation*}
We have the following proposition.	
		\begin{property}\label{prop:characteristic_function_hawkes}
		Let $N$ be a Hawkes process with parameters $(\nu,\phi)$, with $\nu$ a locally integrable non-negative function and $\phi$ a non-negative measurable function such that $\|\phi\|_1<1$. For any continuous function $h$ from $\mathbb{R^{+}}$ into $\mathbb{R}$, 
		\begin{equation*}
		L(h,t) = \mathbb{E}[\emph{exp}((ih* \mathrm{d}N)(t))]
		\end{equation*}
		satisfies
		\begin{equation*}
		L(h,t) = \emph{exp}\big( \int_0^t(C(h,s) -1)\nu(t-s) \mathrm{d}s\big),
		\end{equation*}
		where $C$ is solution of the equation
		\begin{equation*}
		C(h,\cdot) = \emph{exp}\big( ih + (C(h,\cdot)-1)*\phi  \big).
		\end{equation*}
	\end{property}
	\begin{proof}
	Let $\tilde{N}$ be a Hawkes process with parameters $(\phi, \phi)$ and $N^0$ a Poisson process with intensity $\nu$. Let $(\tilde{N}^j)_{j\in\mathbb{N}^{\star}}$ be independent copies of $\tilde{N}$, also independent of $N^0$. Using the population interpretation of Hawkes processes, see Appendix C.1 in \cite{EER17}, we deduce the following equality in law:
		\begin{equation*}
		N_t \overset{\mathcal{L}}{=} N^0_t + \sum_{j=1}^{N^0_t}\tilde{N}^{j}_{t-T_j},
		\end{equation*}
		where $(T_j)_{j\in\mathbb{N}^{\star}}$ are the jump times of the process $N^0$. Consequently
		\begin{equation*}
		(ih*\mathrm{d}N)(t)  \overset{\mathcal{L}}{=} (ih*\mathrm{d}N^0)(t) + \sum_{j=1}^{N^0_t}(ih*\mathrm{d}\tilde{N}^j)(t-T_j).
		\end{equation*}
		Then taking the exponential and conditional expectation with respect to $N^0$ we get
		\begin{eqnarray*}
		\mathbb{E}\big[\text{exp}\big((ih*\mathrm{d}N)(t)\big)|N^0\big] &=& \text{exp}\big((ih*\mathrm{d}N^0)_t\big)\prod_{j=1}^{N^0_t}\tilde{L}(h,t-T_j)\\
		&=&\text{exp}\Big(\big(\big(ih + \text{log}(\tilde{L}(h,\cdot)) \big)*\mathrm{d}N^0)(t)\Big),
		\end{eqnarray*}
		where $\tilde{L}$ is defined as $L$ with $\tilde{N}$ instead of $N$. Remark that 
		\begin{equation*}
		    \big(\big(ih + \text{log}(\tilde{L}(h,\cdot)) \big)*\mathrm{d}N^0)(t) = \sum_{j=1}^{N^0_t}ih(t-T_j) + \text{log}(\tilde{L}(h,t-T_j))
		\end{equation*}
		and that $Re\big(\log\big(\tilde{L}(h,\cdot)\big)\big)\leq0$ since $|  \tilde{L}(h,\cdot)|\leq 1$. Thus using Proposition \ref{lemma:exponential_formula} in Appendix \ref{annex:poisson}, we get
		\begin{equation*}
		L(h,t) = \text{exp}\big( \int_0^t \big(e^{ih(t-s)}\tilde{L}(h, t-s) -1\big)\nu(s) \mathrm{d}s  \big).
		\end{equation*}
In the same way, we have
		\begin{equation*}
		\tilde{L}(h,t) = \text{exp}\big(\int_0^t \big(e^{ih(t-s)}\tilde{L}(h,t-s) -1\big)\phi(s)\mathrm{d}s  \big).
		\end{equation*}
Thus setting
		\begin{equation*}
		C(h,t) = e^{ih(t)} \tilde{L}(h,t),
		\end{equation*}
		we obtain
		\begin{equation*}
		L(h,t) = \text{exp}\big( \int_0^t \big(C(h, s) -1\big)\nu(t-s)\mathrm{d}s  \big)
		\end{equation*}
		and 
		\begin{equation*}
		C(h,\cdot) = \text{exp}\big(ih+ \big(C(h, \cdot) -1\big)*\phi  \big).
		\end{equation*}
	\end{proof}

	\subsubsection{Step $2$}
	\label{proof:rewrite_price}
We consider the price model \eqref{eq:price_kernel}. Let $M^{a,T}$ be defined by
	\begin{equation*}
	    M^{a,T}_t = N^{a,T}_t - \int_{0}^{t}\lambda_s^{a,T}\mathrm{d}s.
	\end{equation*}
	We define $M^{b,T}$ the same way replacing the superscript $a$ by $b$ in the above equation.  We have the following result. 
	\begin{lemma}
		\label{lemma:rewrite_price}
		The price process \eqref{eq:price_kernel} can be written as
		\begin{equation*}
		P^T_t = \big(1 + \int_{0}^{+\infty}\psi^T(v)\mathrm{d}v\big)(M_t^{a,T} - M_t^{b,T}).
		\end{equation*}
	\end{lemma}
	\begin{proof}
		We have
		\begin{equation*}
		P^T_t =  \int_{0}^{t}\big(1 + \int_{0}^{+\infty}\psi^T(v)\mathrm{d}v\big)\big(1 - \int_{0}^{t-u}  \phi^T(v)\mathrm{d}v\big)  \mathrm{d}(N^{a,T}-N^{b,T})_u.
		\end{equation*}
		 We first deal with the term $T_1$ defined by
		\begin{eqnarray*}
			T_1 &=& \int_{0}^{t}\big(1 + \int_{0}^{+\infty}\psi^T(v)\mathrm{d}v\big) \int_{0}^{t-u}  \phi^T(v)\mathrm{d}v  \mathrm{d}(N^{a,T}-N^{b,T})_u\\
			&=&\big(1 + \int_{0}^{+\infty}\psi^T(v)\mathrm{d}v\big) \int_{0}^{t} \int_{u}^{t}  \phi^T(v-u)\mathrm{d}v  \mathrm{d}(N^{a,T}-N^{b,T})_u.
		\end{eqnarray*}
		Using Fubini-Tonelli theorem we get
		\begin{eqnarray*}
			T_1 &=&\big(1 + \int_{0}^{+\infty}\psi^T(v)\mathrm{d}v\big)\int_{0}^{t} \int_{0}^{v}  \phi^T(v-u)  \mathrm{d}(N^{a,T}-N^{b,T})_u \mathrm{d}v.
		\end{eqnarray*}
		Thus we deduce
		\begin{equation*}
		T_1 =\big(1 + \int_{0}^{+\infty}\psi^T(v)\mathrm{d}v\big)\int_{0}^{t}\big( \lambda^{a,T}_v- \mu - \lambda^{b,T}_v +  \mu \big) \mathrm{d}v.
		\end{equation*}
		Finally
		\begin{eqnarray*}
		P^T_t &=& \big(1 + \int_{0}^{+\infty}\psi^T(v)\mathrm{d}v\big) \int_{0}^{t}\big(\mathrm{d}N^{a,T}_v -  \lambda^{a,T}_v\mathrm{d}v - \mathrm{d}N^{b,T}_v+\lambda^{b,T}_v \mathrm{d}v\big)\\ 
		&=& \big(1 + \int_{0}^{+\infty}\psi^T(v)\mathrm{d}v\big)(M^{a,T}_t - M^{b,T}_t).
		\end{eqnarray*}
	\end{proof}
Lemma \ref{lemma:rewrite_price} leads to
\begin{equation*}
\overline{P}^T_t = \frac{1-a^T}{T\mu^T} \big(1 + \int_0^{+\infty}\psi^T(v)\mathrm{d}v\big)(M^{a,T}_{tT} - M^{b,T}_{tT}).
\end{equation*}

	\subsubsection{Step $3$} We temporarily drop the superscripts $a$ and $b$. Indeed, the results are valid both for buy and sell order flows.
	\label{proof:results_scaling_limit_hawkes}
	Consider the sequences
	\begin{equation}
	\label{eq:proof_seq}
	X^T_t = \frac{1-a^T}{T\mu^T}N^T_{tT}, \,\,\,\Lambda^T_t =  \frac{1-a^T}{T\mu^T}\int_{0}^{tT}\lambda^T_s\mathrm{d}s, \,\,\, Z^T_t = \sqrt{\frac{T\mu^T}{1-a^T}}\big(X^T_t-\Lambda^T_t\big).
	\end{equation}
The following result is borrowed from \cite{JR16b}.
	\begin{property}
		\label{prop:tightness_hawkes}	
	The sequence $(\Lambda^T, X^T,Z^T)$ is tight. Furthermore, for any limit point $(\Lambda,X,Z)$ of $(\Lambda^T, X^T,Z^T)$, $Z$ is a continuous martingale, $\left[Z,Z\right] = X$ and $\Lambda = X$.
	\end{property}
In addition, we have the following proposition which extends Theorem 3.1 in \cite{JR16b}.
	\begin{property}
		\label{th:scaling_limit_multi_hawkes}
Under Assumptions \ref{assumption:market_impact}	and \ref{assumption:structure}, for any limit point $(X, Z)$ of $(X^T, Z^T)$, there exists a Brownian motion $B$ on $(\Omega, \mathcal{A}, \mathbb{P})$ (up to extension of the space) such that
		\begin{equation*}
		Z_t = B_{X_t}
		\end{equation*}
		and $X$ is a solution of the stochastic rough Volterra equation
		\begin{equation}\label{eqxrough}
		X_t = \int_{0}^{t}F^{\alpha, \lambda}(t-s) \mathrm{d}s + \frac{1}{\sqrt{\delta \lambda}}\int_{0}^{t}F^{\alpha, \lambda}(t-s) \mathrm{d}B_{X_s}.
		\end{equation}
		Moreover, for any $\varepsilon>0$, the process $X$ has H\"older regularity $1\wedge(2\alpha - \varepsilon)$.
	\end{property}
	Note that we are here under more general assumptions than in Theorem 3.1 in \cite{JR16b}. Indeed in \cite{JR16b} we have
	\begin{equation*}
	\int_{t}^{+\infty}\phi(s)\mathrm{d}s = K t^{-\alpha},
	\end{equation*} 
	while we only know that
	\begin{equation*}
	\int_0^t \int_{s}^{+\infty}\phi(u)\mathrm{d}u\mathrm{d}s =  L(t) t^{1-\alpha},
	\end{equation*}
	with $L$ a slowly varying function. To prove Proposition \ref{th:scaling_limit_multi_hawkes}, it is enough to get the following lemma. The rest of the proof is similar to that in \cite{JR16b}.
	\begin{lemma}
		\label{lemma:convergence_mittag}
		The sequence of functions $\rho^T(t) = \frac{1-a^T}{a^T}\psi^T(Tt)T$ converges weakly towards $f^{\alpha, \lambda}$. Furthermore $\int_0^t \rho^T(s)\mathrm{d}s$ converges uniformly towards $F^{\alpha, \lambda}$.
	\end{lemma}
	\begin{proof}
		 The function $\rho^T$ is non-negative with integral equal to one. So it can be interpreted as the density of a random variable. Hence it is enough to show that its Laplace transform converges pointwise to get weak convergence. We have for $z>0$
	$$
			\hat{\rho}^T(z) = \frac{1-a^T}{a^T}\hat{\psi}^T(\frac{z}{T})
			=\frac{\hat{\phi}(\frac{z}{T})}{1-a^T(1-a^T)^{-1}(\hat{\phi}(\frac{z}{T})-1)}.
		$$
Let
		\begin{equation*}
			R(t) = \int_0^t \int_{s}^{+\infty}\phi(u)\mathrm{du}\mathrm{d}s.
		\end{equation*} 
Recall that from Theorem \ref{th:market_impact_limit}, $R(t) = t^{1-\alpha}L(t)$. By Karamata's Tauberian theorem, see Theorem \ref{th:karamata} in Appendix \ref{annex:karamata}, we have
		\begin{equation*}
			\hat{R}(z)\sim_{0^+} z^{\alpha-2}L(\frac{1}{z})\Gamma(2-\alpha).
		\end{equation*}
Integrating by parts twice we obtain
		$$
			\hat{R}(z) = \int_{0}^{+\infty}e^{-zs}R(s)\mathrm{d}s=\frac{1}{z^2}\big(1-\hat{\phi}(z) \big).
		$$		So we get
		\begin{equation*}
			a^T(1-a^T)^{-1}(1-\hat{\phi}(\frac{z}{T}))\sim_{T\rightarrow+\infty} a^T (1-a^T)^{-1}T^{-\alpha} L(T)\frac{L(\frac{T}{z})}{L(T)} z^{\alpha} \Gamma(2-\alpha).
		\end{equation*}
		We have shown in Theorem \ref{th:market_impact_limit} that
		\begin{equation*}
		    a^T(1-a^T)^{-1}T^{-\alpha}L(T)\underset{T\rightarrow+\infty}{\rightarrow}K.
		\end{equation*} 
		Since $L$ is a slowly varying function, see Appendix \ref{annex:karamata}, we deduce
		\begin{equation*}
			\underset{T\rightarrow+\infty}{\lim} a^T(1-a^T)^{-1}(1-\hat{\phi}(\frac{z}{T}))=z^{\alpha}\Gamma(2-\alpha)K,
		\end{equation*}
		and finally
		\begin{equation*}
			\underset{T\rightarrow+\infty}{\lim}\hat{\rho}^T(z)=\frac{1}{1+K\Gamma(2-\alpha)z^{\alpha}} = \frac{\lambda}{\lambda+z^{\alpha}}=\hat{f}^{\alpha, \lambda},
		\end{equation*}
		with $\lambda = (K\Gamma(2-\alpha))^{-1}$. 
	The uniform convergence in Lemma \ref{lemma:convergence_mittag} is obviously deduced from Dini's theorem.
	\end{proof}
	
	We finally show that the sequence $(X^T,Z^T)_{T\geq 0}$ converges in law for the Skorokhod topology. We already know that it is tight, so it is enough to prove that all the limit points have the same law.\\
	
Let $(X,Z)$ be a limit point of $(X^T,Z^T)_{T\geq 0}$. Using Proposition \ref{th:scaling_limit_multi_hawkes} together with the stochastic Fubini theorem, see \cite{veraar2012stochastic}, we have
	\begin{equation*}
	    X_t = \int_0^t f^{\alpha, \lambda}(t-s) \big(s + \frac{1}{\sqrt{\delta\lambda}} Z_s\big) \mathrm{d}s.
	\end{equation*}
	From Example 42.2 in \cite{SKM93}, this leads to
	\begin{equation*}
	    D^{\alpha}X_t +\lambda X_t - \lambda t = \sqrt{\frac{\lambda}{\delta}} Z_t,
	\end{equation*}
	where $D^{\alpha}$ is the fractional derivative operator defined in Appendix \ref{annex:fractional_calculus}. Thus the law of $(X,Z)$ is uniquely determined by the law of $X$.  Consequently it is enough to prove uniqueness in law for limit points of $(X^T)_{T\geq 0}$ to get convergence in law of $(X^T,Z^T)_{T\geq 0}$. For this we prove that the characteristic function of any limit point $X$ of the sequence $(X^T)_{T\geq0}$ is a functional of the solution of a fractional Riccati equation. Uniqueness in law is then a consequence from the uniqueness of the solution of this equation.
	\begin{property}
		\label{prop:characteristic_function_vol}
		Let $X$ be a limit point of $(X^T)_{T\geq0}$ and $h$ a continuously differentiable function from $\mathbb{R}^{+}$ to $\mathbb{R}$ such that $h(0)=0$. The function
		\begin{equation*}
		K(h, t) = \mathbb{E}[\emph{exp}(ih* \mathrm{d}X)_t]
		\end{equation*}
		satisfies
		\begin{equation*}
		K(h, t) = \emph{exp}( \int_{0}^{t}g(s) \mathrm{d}s),
		\end{equation*}
		with $g$ the unique continuous solution of the rough Volterra Riccati equation
		\begin{equation}
		\label{eq:proof_riccati_a}
		g = f^{\alpha, \lambda}*\big(\delta^{-1}\frac{1}{2}g^2 + ih\big).
		\end{equation}
	\end{property}
	To show this result, we are inspired by the methodology of \cite{EER16}. However, note again that we are in a more general setting.
	\begin{proof}
	    Recall that
		\begin{equation*}
		X^T_t = \frac{1-a^T}{T\mu^T}N^T_{tT}.
		\end{equation*}
		We introduce the following quantities:
		\begin{equation*}
		h^T(t) = \frac{1-a^T}{T\mu^T}h(\frac{t}{T}),~ L^T(h^{T},t) = \mathbb{E}[\text{exp}(ih^{T}* \mathrm{d}N^T)(t)]~\mathrm{and}~K^T = L^T(h^T, tT).
		\end{equation*}
		For every $T$, according to Proposition \ref{prop:characteristic_function_hawkes}, there exists a function $C^T$ solution of
		\begin{equation*}
		C^T = \text{exp}\big(ih^T+ \big(C^T -1\big)*\phi^T  \big)
		\end{equation*}
		such that
		\begin{equation*}
		L^T(h^T,t) = \text{exp}\Big( \int_0^t\big(C^T(s) -1\big)\mu^T \mathrm{d}s \Big).
		\end{equation*}
		Now define the sequence $g^T$
		\begin{equation*}
		g^T(s) = C^T(sT)-1.
		\end{equation*}
		We have
		\begin{equation}
		\label{eq:proof_car_vol_a}
		K^T= \text{exp}\Big( \frac{g^T}{1-a^T}*(T (1-a^T) \mu^T\mathbf{1}_{\mathbb{R}^+}) \Big)\,\, \mathrm{and}\,\,\, g^T +1= \text{exp}\big(\frac{1-a^T}{T\mu^T}ih +  g^T* (T\phi^T(\cdot T)) \big).
		\end{equation}
		An immediate adaptation of Proposition 6.4. in \cite{EER16} gives that for any $s\in[0,t]$
		\begin{equation}
		\label{eq:proof_riccati_ineq}
		|g^T(s)|\leq c(h)(1-a^T),
		\end{equation}
		with $c(h)$ a positive constant depending only on $h$. Hence for $T$ large enough we have
		\begin{equation}
		\label{eq:proof_riccati_log}
		\log(1+g^T) = g^T - \frac{1}{2}(g^T)^2 - \epsilon^T(h, \cdot),
		\end{equation}
		with $|\epsilon^T(h, \cdot)|\leq c(h) (1-a^T)^3$. According to Equations \eqref{eq:proof_car_vol_a} and \eqref{eq:proof_riccati_log}, we get for every $s\in [0,t]$
		\begin{equation*}
		g^{T}(s) = \frac{1}{2}g^{T}(s)^2 + \epsilon^{T}(h, s) + \frac{1-a^{T}}{{T}\mu^{T}}ih(s) +  g^{T}* \big({T}\phi^{T}(\cdot {T}) \big)(s).
		\end{equation*}
	Using that		\begin{equation*}
		\sum_{i\geq 1} (T\phi^T(\cdot T))^{*i} = T\psi^T(\cdot T),
		\end{equation*}
		we deduce from Lemma 4.1 in \cite{jaisson2015limit} that
		\begin{equation*}
		g^T(s) = \big(T\psi^T(\cdot T)\big)*\big(\frac{1}{2}(g^T)^2 + \epsilon^T(h, \cdot) + \frac{1-a^T}{T\mu^T}ih\big)(s) + \frac{1}{2}g^{T}(s)^2 + \epsilon^{T}(h, s) + \frac{1-a^{T}}{{T}\mu^{T}}ih(s).
		\end{equation*}
		Consequently, letting $\theta_T =(1-a^T)^{-1}g^T$ 
		\begin{equation*}
		\theta_T(s) = \big(T(1-a^T)\psi^T(\cdot T)\big)*\big(\frac{1}{2}\theta_T^2 + \frac{1}{\delta}ih\big)(s) + r_1^T(s),
		\end{equation*}
		with 
		\begin{eqnarray*}
		r_1^T(s)  &=&\big(T(1-a^T)\psi^T(\cdot T)\big)*\big(\epsilon^T(h, \cdot)(1-a^T)^{-2} + (\frac{1}{T(1-a^T)\mu^T}-\delta^{-1})ih\big)(s)\\ &+& (1-a^T)^{-1}\frac{1}{2}(g^{T}(s))^2 + (1-a^T)^{-1}\epsilon^{T}(h, s) + \frac{1}{{T}\mu^{T}}ih(s) . 
		\end{eqnarray*}
Since $a^T$ goes to $1$, we know from Lemma \ref{lemma:convergence_mittag} that in the sense of weak convergence
		\begin{equation*}
		T(1-a^T)\psi^T(\cdot T) \underset{T\rightarrow+\infty}{\rightarrow} f^{\alpha, \lambda}.
		\end{equation*}
		Finally we have
		\begin{equation*}
		\theta_T = f^{\alpha, \lambda}*\big(\frac{1}{2}\theta_T^2 + \frac{1}{\delta}ih\big) + r^T_1 + r_2^T,
		\end{equation*}
		where
		\begin{equation*}
		r_2^T = \big(T(1-a^T)\psi^T(\cdot T) - f^{\alpha, \lambda}\big)*\big(\frac{1}{2}\theta_T^2 + \frac{1}{\delta}ih\big).
		\end{equation*}
		We now prove that $(r_1^T)_{T\geq 0}$ and $(r_2^T)_{T\geq 0}$ goes to $0$ in $C^0([0,t], \mathbb{R})$ for the sup-norm.\\
		
		Using Assumption \ref{assumption:structure}, the second part of Lemma \ref{lemma:convergence_mittag} and Equation \eqref{eq:proof_riccati_ineq}, we get that $(r^T_1)_{T\geq 0}$ goes to zero in $C^0([0,t], \mathbb{R})$. The sequence $(\theta_T)_{T\geq 0}$ is bounded for the sup-norm according to Equation \eqref{eq:proof_riccati_ineq}. Moreover according to Lemma \ref{proof:lemma:theta_compact} (see after the proof)  $\theta_T$ is differentiable, and $(\theta_T')_{T\geq 0}$ is bounded for the sup-norm. By integration by parts we have
		\begin{equation*}
		r_2^T(t) = \big(\int_0^{\cdot}T(1-a^T)\psi^T(s T)\mathrm{d}s - F^{\alpha, \lambda}\big)*\big(\theta_T'\theta_T + \frac{1}{\delta}ih'\big)(t),
		\end{equation*}
		where we have used the fact that $\theta_T(0)=0$ and $h(0)=0$. We then conclude that $(r_2^T)_{T\geq0}$ converges towards $0$ in $C^0([0,t], \mathbb{R})$ using dominated convergence. Lemma \ref{proof:lemma:theta_compact} together with the Ascoli theorem gives that the sequence $(\theta_T)_{T\geq0}$ is relatively compact in $\big(C^0([0,t], \mathbb{R}), \| ~\|_{\infty}\big)$. Moreover for any limit point $\theta$ of the sequence $(\theta_T)_{T\geq0}$, we have that $\theta$ is solution of:
		\begin{equation*}
		\theta = f^{\alpha, \lambda}*\big(\frac{1}{2}\theta^2 + \frac{1}{\delta}ih\big).
		\end{equation*}
		The above equation has a unique continuous solution in $C^0([0,t], \mathbb{R})$, see Section 6.2.4 in \cite{EER16}. Thus the sequence $(\theta_T)_{T\geq0}$ converges toward this solution.\\
		
		Finally remark that
		\begin{equation*}
		    K^T(t) = \mathbb{E}[ \text{exp}(i h*\mathrm{d}X^T)(t) ].
		\end{equation*}
		Thus convergence in law of $(X^T)_{T\geq 0}$ towards $X$ implies that $(K^T)_{T\geq0}$ converges pointwise towards the function $K$. Passing to the limit in (\ref{eq:proof_car_vol_a}) we get
		\begin{equation*}
		K(t)= \text{exp}\Big( (\theta*(\delta \mathbf{1}_{\mathbb{R}^+} ))_t \Big) = \text{exp}\Big( \delta\int_{0}^{t}\theta(s)\mathrm{d}s \Big).
		\end{equation*}

		Letting $g = \delta \theta$, we have the result.
		\end{proof}
	
It is enough to characterize the law of $X$ to know $K(h, t)$ for any $t\in \mathbb{R^+}$ and $h\in C^1_0([0,t], \mathbb{R})$. Therefore uniqueness in law for the limit points of $(X^T)_{T\geq0}$ is a corollary from the uniqueness of continuous solution for the Volterra Riccati Equation (\ref{eq:proof_riccati_a}), see Section 6.2.4 in \cite{EER16}.\\

    We now give the lemma we used in the proof of Proposition \ref{prop:characteristic_function_vol}.
	\begin{lemma}
	    \label{proof:lemma:theta_compact}
	    The functions $(\theta_T)_{T\geq0}$ are continuously differentiable and $(\theta_T')_{T\geq0}$ is bounded in $C^0([0,t], \mathbb{R})$.
	\end{lemma}
	\begin{proof}
    Using the proof of Proposition \ref{prop:characteristic_function_hawkes} we have
	   \begin{equation*}
	        \theta_T = (1-a^T)^{-1} \Big( \mathbb{E}[\text{exp}\big((ih+ih*\mathrm{d}\tilde{N}_{\cdot T}^T)\frac{1-a^T}{T\mu^T}\big)]-1 \Big),
	    \end{equation*}
	    with $\tilde{N}$ a Hawkes processes with parameters $(\phi^T, \phi^T)$ where $\phi^T = a^T\phi$. Since $h(0)=0$, $h*\mathrm{d}\tilde{N}_{\cdot T}^T$ admits a derivative and for any $s\in [0,t]$
	    \begin{equation*}
	        (h*\mathrm{d}\tilde{N}_{\cdot T}^T)'(s) = (h'*\mathrm{d}\tilde{N}_{\cdot T}^T)(s).
	    \end{equation*}
	    Furthermore we have
	    \begin{equation*}
	        | (h'*\mathrm{d}\tilde{N}_{\cdot T}^T)(s) |\leq \| h'\|_{\infty} \tilde{N}_{tT}^T.
	    \end{equation*}
	    Using that
	    \begin{equation*}
	    \tilde{\lambda}_s^T = \psi^T(s) + \int_0^t \psi^T(t-s)\mathrm{d}\tilde{M}^T_s,
	    \end{equation*}
	    we get
$$
	        (1-a^T)\mathbb{E}[\tilde{N}^T_{tT}]\leq(1-a^T)\mathbb{E}[\int_0^{tT}\tilde{\lambda}^T_s\mathrm{d}s]\\
	        \leq \int_0^{tT}(1-a^T)\psi^T(s)\mathrm{d}s\\
	       \leq 1.
	$$
	    Consequently using derivation for integral with parameters, $\theta^T$ is differentiable and
	    \begin{equation*}
	        \theta_T' = (1-a^T)^{-1}\mathbb{E}[\big( ih' + ih'*\mathrm{d}\tilde{N}_{\cdot T}^T \big)\frac{1-a^T}{T\mu^T} \text{exp}\big((ih+ih*\mathrm{d}\tilde{N}_{\cdot T}^T)\frac{1-a^T}{T\mu^T}\big)].
	    \end{equation*}
Thus we have for all $s\in [0,t]$
	    \begin{equation*}
	        |\theta_T'(s)| \leq \frac{1}{T\mu^T(1-a^T)}(1-a^T)\mathbb{E}[\|h'\|_{\infty} + \|h'\|_{\infty} \tilde{N}_{tT}^T].
	    \end{equation*}
	    The right hand side is finite and independent of $s$, 
	    consequently the sequence $(\theta_T')_{T\geq 0}$ is bounded in $C^0([0,t], \mathbb{R})$.
	    \end{proof}
	Finally we have proved that the sequence $(X^T,Z^T)_{T\geq 0}$ converges in law for the Skorokhod topology.

	\subsubsection{Step $4$}
	\label{subsubsec:consequences_price_scaling_limit}
	Consider the sequence $(X^{a,T},Z^{a,T})_{T\geq 0}$ (resp. $(X^{b,T},Z^{b,T})_{T\geq 0}$) defined the same way as in Equation \eqref{eq:proof_seq} with $(N^{a,T})_{T\geq 0}$ (resp. $(N^{b,T})_{T\geq 0}$) instead of $(N^T)_{T\geq 0}$. According to Lemma \ref{lemma:rewrite_price} we have
$$
        \overline{P}^T_t = \frac{1-a^T}{T\mu^T} \big(1 + \int_0^{+\infty}\psi^T(u)\mathrm{d}u\big)\big(M^{a,T}_{tT} - M^{b,T}_{tT}\big)
= \frac{1}{T\mu^T} \big(M^{a,T}_{tT} - M^{b,T}_{tT}\big).
$$
	Thus,
	\begin{equation*}
	 	\overline{P}^T_t = \frac{1}{\sqrt{T\mu^T(1-a^T)}} \big(Z^{a,T}_{t} - Z^{b,T}_{t}  \big).
	\end{equation*}
Using Step 3, we get that $(Z^{a,T})_{T\geq0}$, and $(Z^{b,T})_{T\geq0}$ converge for the Skorohod topology. These sequences being independent,   $(\overline{P}^{T})_{T\geq0}$ converges towards a process $\widehat{P}$ in the Skorokhod topology. Furthermore we deduce from Proposition \ref{th:scaling_limit_multi_hawkes} together with Assumption \ref{assumption:structure} that there exist two independent Brownian motions $B^a$ and $B^b$ such that
	\begin{equation*}
	\widehat{P}_t = \frac{1}{\sqrt{\delta}}\big(B^a_{X_t^a} - B^b_{X_t^b}\big),
	\end{equation*}
	where $X^a$ (resp. $X^b$) is the limit of the sequence $(X^{a,T})_{T\geq 0}$ (resp. $(X^{a,T})_{T\geq 0}$) and is solution of Equation \eqref{eqxrough} with Brownian motion $B^a$ (resp. $B^b$). Hence $X = \frac{X^a + X^b}{\delta}$ is solution of
	\begin{equation*}
	X_t = \frac{2}{\delta}\int_{0}^{t}F^{\alpha, \lambda}(t-s) \mathrm{d}s + \frac{1}{\delta\sqrt{\lambda}}\int_{0}^{t}F^{\alpha, \lambda}(t-s)\frac{1}{\sqrt{\delta}}\mathrm{d}\big(B^a_{X^a_s}  + B^b_{X^b_s }\big).
	\end{equation*}
	Moreover there exists a Brownian motion $W$ such that $W_{X_t} = \frac{1}{\sqrt{\delta}}(B^a_{X^a_t} + B^b_{X^b_t})$. Consequently
	\begin{equation*}
	X_t = \frac{2}{\delta}\int_{0}^{t}F^{\alpha, \lambda}(t-s) \mathrm{d}s + \frac{1}{\delta\sqrt{ \lambda}}\int_{0}^{t}F^{\alpha, \lambda}(t-s)\mathrm{d}W_{X_s }.
	\end{equation*}

	\subsubsection{Step $5$}
	\label{proof:regularity}
	We first recall a result from \cite{JR16b}.
	\begin{property}
		\label{prop:volterra_regularity}
	Let $X$ be a solution of the stochastic Volterra equation \eqref{eq:price_var_eq}. Then for any $\varepsilon>0$, almost surely, $X$ has H\"older regularity $1\wedge (2\alpha-\varepsilon)$. And if $\alpha>1/2$, $X$ is almost surely differentiable.
	\end{property}
	We now give a new result on the regularity of the solution of Equation \eqref{eq:price_var_eq}.
	\begin{property}
	\label{lemma:regularity_var}
Let $\alpha\leq\frac{1}{2}$. Let $X$ be a solution of the stochastic Volterra equation \eqref{eq:price_var_eq}. Then, almost surely, $X$ is not continuously differentiable.
	\end{property}
	\begin{proof}
As already seen in Step 3, $X$ satisfies
	\begin{equation}
	\label{eq:regularity_a}
	    D^{\alpha}X_t = -\lambda X_t + \frac{2\lambda}{\delta}t + \frac{\sqrt{\lambda}}{\delta} W_{X_t}.
	\end{equation}
Applying the law of iterated logarithm we get for $0\leq t\leq 1$
	\begin{equation*}
	    \underset{s\rightarrow t^-}{\limsup}  \frac{D^{\alpha}X_t - D^{\alpha}X_s  - \frac{2\lambda}{\delta}(t-s)}{\sqrt{2(X_t - X_s)\log\log\big((X_t - X_s)^{-1}\big)}}  = \frac{\sqrt{\lambda}}{\delta}.
	\end{equation*}
Assume that $X$ is continuously differentiable. According to Appendix \ref{annex:fractional_calculus} we have
	\begin{equation*}
	    D^{\alpha}X_t = \frac{1}{\Gamma(1-\alpha)}\int_0^t (t-s)^{-\alpha}X_s'\mathrm{d}s.
	\end{equation*}
	Let $t$ be such that $X'_t \neq 0$. Such a point almost surely exists because $X$ is not constant. Indeed suppose it is constant, as $X_0 = 0$ it implies that $X=0$. But obviously the null function is not solution of Equation \eqref{eq:regularity_a}. For such $t$ using that
	\begin{equation*}
	    X_t - X_s \sim_{s\rightarrow t}(t-s)X'_t,
	\end{equation*}
	we have
	\begin{equation*}
	    \underset{s\rightarrow t^-}{\lim}  \frac{t-s}{\sqrt{2(X_t - X_s)\log\log\left[(X_t - X_s)^{-1}\right]}}  = 0.
	\end{equation*}
	Hence
	\begin{equation}
	\label{eq:proof_regularity}
	    \underset{s\rightarrow t^-}{\limsup}  \frac{D^{\alpha}X_t - D^{\alpha}X_s  }{\sqrt{2(X_t - X_s)\log\log\left[(X_t - X_s)^{-1}\right]}}  = \frac{\sqrt{\lambda}}{\delta}.
	\end{equation}
We now give a bound on	$|D^{\alpha}X_t - D^{\alpha}X_s|$, for $s<t$, where $\| X'\|_{\infty}$ denotes the supremum norm of $X'$:
	\begin{eqnarray*}
	    |D^{\alpha}X_t - D^{\alpha}X_s|  &=& \big|\int_0^s \big((t-u)^{-\alpha} - (s-u)^{-\alpha} \big)X'_s\mathrm{d}u + \int_s^t(t-u)^{-\alpha} X'_u\mathrm{d}u\big|\\
	    &\leq & \int_0^s \left|(t-u)^{-\alpha} - (s-u)^{-\alpha} \right| \| X'\|_{\infty}\mathrm{d}u + \frac{\| X'\|_{\infty}}{1-\alpha}(t-s)^{1-\alpha}\\
	    &\leq & \left|\int_0^s (t-u)^{-\alpha} - (s-u)^{-\alpha} \right| \mathrm{d}u\| X'\|_{\infty} + \frac{\| X'\|_{\infty}}{1-\alpha}(t-s)^{1-\alpha}\\
	    &\leq & \big(\int_0^{t-s} u^{-\alpha}\mathrm{d}u + \int_s^t u^{-\alpha} \mathrm{d}u \big) \| X'\|_{\infty} + \frac{\| X'\|_{\infty}}{1-\alpha}(t-s)^{1-\alpha}\\
	    &\leq & \big(\frac{1}{1-\alpha}(t-s)^{1-\alpha} + (t-s) s^{-\alpha}\big) \| X'\|_{\infty}+ \frac{\| X'\|_{\infty}}{1-\alpha}(t-s)^{1-\alpha}.
	\end{eqnarray*}
	We get
	\begin{equation*}
	    \underset{s\rightarrow t^-}{\lim}  \frac{D^{\alpha}X_t - D^{\alpha}X_s  }{\sqrt{2(X_t - X_s)\log\log\left[(X_t - X_s)^{-1}\right]}}  = 0.
	\end{equation*}
	This is in contradiction with Equation \eqref{eq:proof_regularity}, hence $X$ cannot be continuously differentiable.
	\end{proof}
	    
	\subsection{Proof of Theorem \ref{th:characteristic_function_variance}}
	\label{proof:th:characteristic_function_variance}    
	We have seen in Section \ref{subsubsec:consequences_price_scaling_limit} that $X = (X^a + X^b)/\delta$, with $X^a$ and $X^b$ independent copies of the limit of the sequence $(X^T)_{T\geq 0}$. From Proposition \ref{prop:characteristic_function_vol}, we immediately obtain Theorem \ref{th:characteristic_function_variance}.
	    
	\appendix
	\section{Appendix}
	\label{sec:annex}
	
	\subsection{Mittag-Leffler functions}
	\label{annex:mittag}
	Let $(\alpha, \beta)\in(\mathbb{R}^{+}_{\star})^2$. The Mittag-Leffler function $E_{\alpha, \beta}$ is defined for $z\in\mathbb{C}$ by\begin{equation*}
		E_{\alpha, \beta}(z) = \sum_{n\geq 0}\frac{z^n}{\Gamma(\alpha n+\beta)}.
	\end{equation*}
	For $(\alpha, \lambda)\in(0,1]\times\mathbb{R}^+$, we also define 
	\begin{eqnarray*}
		f^{\alpha, \lambda}(t)&=&\lambda t^{\alpha-1} E_{\alpha, \alpha}(-\lambda t^{\alpha}), \,\, t>0,\\
		F^{\alpha, \lambda}(t) &=&  \int_{0}^{t}f^{\alpha, \lambda}(s)\mathrm{d}s, \,\, t\geq0.
	\end{eqnarray*}
	The function $f^{\alpha, \lambda}$ is a density function on $\mathbb{R}_+$ called the Mittag-Leffler density function. Its Laplace transform is
	\begin{equation*}
		\hat{f}^{\alpha, \lambda}(z) = \frac{\lambda}{\lambda+z^{\alpha}}.
		\end{equation*}
	When $\alpha = 1$, the Mittag-Leffler density simply corresponds to the exponential law with parameter $\lambda$.

	\subsection{Tauberian theorems}
	\label{annex:karamata}
	The following results can be found in \cite{HT89}.
	\begin{definition}
		\label{def:slowly_varying function}
		A measurable function $L:\mathbb{R}^+\rightarrow \mathbb{R}$ is slowly varying if for all $s>0$
		\begin{equation*}
			\frac{L(st)}{L(t)}\underset{t\rightarrow+\infty}{\rightarrow} 1.
		\end{equation*}
	\end{definition}
	
	\begin{property}
	\label{prop:annex_slowly_varying}
	    Let $L$ be a slowly varying function and $\alpha>0$, then
	    \begin{equation*}
	        t^{-\alpha}L(t) \underset{t\rightarrow+\infty}{\rightarrow} 0.
	    \end{equation*}
	\end{property}
	
	\begin{theorem}
	\label{th:characterisation}(Characterisation theorem)
		Let $U$ be a positive measurable function on $\mathbb{R}_+$ such that for all $s\in C$, with $C$ a set with positive Lebesgue measure
		\begin{equation*}
		\frac{U(ts)}{U(t)}\underset{t\rightarrow+\infty}{\rightarrow}g(s)>0,
		\end{equation*}
		for some function $g$. Then the previous limit can be extended for all $s>0$. Let $\tilde{g}$ be this limiting function extending $g$. There exist $\alpha\in\mathbb{R}$ such that $g(t) = t^{\alpha}$ and a slowly varying function $L$ such that $U(t) = t^{\alpha}L(t)$.\\
	\end{theorem}

	\begin{theorem}
	\label{th:karamata}(Karamata's Tauberian theorem) Let $U$ be a measurable non-negative function, $c\geq 0$, $\rho>-1$ and assume $\hat{U}(z) = \int_0^{+\infty}e^{-zs}U(s)\mathrm{d}s$ is finite for any $z>0$. Then
	\begin{equation*}
	    U(t)\sim_{+\infty}ct^{\rho}\frac{L(t)}{\Gamma(1+\rho)}
	\end{equation*}
	for $L$ a slowly varying function implies
	\begin{equation*}
	    \hat{U}(z)\sim_{0^+}cz^{-\rho-1}L(\frac{1}{z}).
	\end{equation*}
	
	\end{theorem}

	\subsection{Fractional derivative}
    \label{annex:fractional_calculus}
	For $\alpha\in [0,1)$, the fractional derivative operator $D^{\alpha}$ is defined for $h$ $\lambda$-H\"older function (with $\lambda>\alpha$) by
	\begin{equation*}
	    D^{\alpha}f(t) = \frac{1}{\Gamma(1-\alpha)}\frac{\mathrm{d}}{\mathrm{d}t}\int_0^t (t-s)^{-\alpha}h(s)\mathrm{d}s.
	\end{equation*}
	Note that if the function $h$ is continuously differentiable and $f(0) =0$. The derivation for integral with parameters gives
    \begin{equation*}
	    D^{\alpha}f(t) = \frac{1}{\Gamma(1-\alpha)}\int_0^t (t-s)^{-\alpha}f'(s)\mathrm{d}s.
	\end{equation*}
	More information on fractional differential operator can be found in \cite{SKM93}.
	
	\subsection{A result on inhomogenous Poisson process}
	\label{annex:poisson}
We recall the following well known result.
	\begin{property}
	\label{lemma:exponential_formula}(Exponential formula)
	Let $N$ be an inhomogenous Poisson process with intensity $\nu$ and $f$ be a complex measurable function defined on $\mathbb{R}^{+}$ such that $Re(f)\leq 0$. Consider the function 
	\begin{equation*}
	    N_{f}(t) = \sum_{i=1}^{N_t}f(T_i),
	\end{equation*}
    where $(T_i)_{i\in \mathbb{N}}$ are the jump times of $N$. For any $t\geq 0$ we have
    \begin{equation*}
        \mathbb{E}[\emph{exp}\big(N_{f}(t)\big)] = \emph{exp}\big(\int_0^t (e^{f(s)}-1)\nu(s) \mathrm{d}s\big).
    \end{equation*}
	\end{property}
		
	\bibliographystyle{plain}
	\bibliography{biblioJusRos}
	
\end{document}